\documentclass{article}
\usepackage{amsfonts}
\usepackage{amsmath}

\newtheorem{theorem}{Theorem}

\newtheorem{lemma}[theorem]{Lemma}

\newtheorem{proposition}[theorem]{Proposition}
\newtheorem{remark}[theorem]{Remark}

\newenvironment{proof}[1][Proof]{\noindent\textbf{#1.} }{\ \rule{0.5em}{0.5em}}
\input{tcilatex}
\begin{document}

\title{Reverse test and quantum analogue of classical fidelity and
generalized fidelity}
\author{Keiji Matsumoto \\
National Institute of Informatics, 2-1-2, Hitotsubashi, Chiyoda-ku, Tokyo \\
Quantum Computation and Information Project, SORST, JST,\\
5-28-3, Hongo, Bunkyo-ku, Tokyo 113-0033, Japan}
\maketitle

\section{Introduction}

In doing hypothesis test in quantum mechanical setting, key part is choice
of measurement which maps given pair $\left\{ \rho ,\sigma \right\} $ of
quantum states to a pair $\left\{ p,q\right\} $ of probability
distributions. Its inverse operation, or a CPTP map form $\left\{
p,q\right\} $ to $\left\{ \rho ,\sigma \right\} $ is called \textit{reverse
test, }and plays an essential role in characterizing largest monotone
quantum analogue of relative entropy\thinspace \cite{Matsumoto-2}\cite%
{Matsumoto-3}. In this paper, we exploit the same line of argument in
studying quantum analogues of \ affinity, or classical fidelity $F\left(
p,q\right) =\sum_{x}\sqrt{p\left( x\right) }\sqrt{q\left( x\right) }$, and
more \textit{generalized fidelity }$F_{f}\left( p,q\right) :=\sum_{x}p\left(
x\right) f\left( q\left( x\right) /p\left( x\right) \right) $, where $f$ is
an operator monotone function on $[0,\infty )$. (For example, $f\left(
t\right) =x^{\alpha }$ ($0<\alpha <1$.) 

In the paper, based on reverse test, we define $F_{\min }\left( \rho ,\sigma
\right) $, which turns out to equal $\mathrm{tr}\,\rho \sqrt{\rho
^{-1/2}\sigma \rho ^{-1/2}}$. \ This quantity is monotone increasing by the
application of TPCP maps, and in fact is the smallest one among the numbers
satisfying these properties, while $F\left( \rho ,\sigma \right) $ is the
largest. It is also proved that  $F_{\min }$ satisfies strong joint
concavity using reverse test. 

For  generalized fidelity, we also introduce  $F_{f}^{\min }\left( \rho
,\sigma \right) $ in the similar manner, which turns out to equal $\mathrm{tr%
}\,\rho f\left( \rho ^{-1/2}\sigma \rho ^{-1/2}\right) $. Again, this
quantity is monotone increasing and is the smallest one among the numbers
satisfying these properties. Joint concavity of $F_{f}^{\min }\left( \rho
,\sigma \right) $ is also proved using reverse test.      

It is known that fidelity between infinitesimally different states gives
rise to SLD Fisher information metric $J^{S}$, or the smallest monotone
metric, and that $\cos ^{-1}F\left( \rho ,\sigma \right) $ equals the
integral of $J^{S}$ along the geodesic, or the curve which minimize the
integral, connecting $\rho $ and $\sigma $.

Correspondingly, $F_{\min }\left( \rho ,\sigma \right) $ gives rise to RLD
Fisher information metric $J^{R}$, or the largest monotone metric. However,
the integral of RLD Fisher information metric along the geodesic does not
equal $\cos ^{-1}F_{\min }\left( \rho ,\sigma \right) $. In fact, $\ $cosine
of the integral, denoted by $F_{R}\left( \rho ,\sigma \right) $, is another
monotone quantum analogue of classical fidelity, and is the smallest one
among those which satisfy triangle inequality. On the other hand, $F_{\min
}\left( \rho ,\sigma \right) $ is the integral of RLD Fisher information
along the curve which minimize the integral for all the curves with
commutative RLD.

An upper and a lower bound of the quantum statistical distance $\Delta
\left( \rho ,\sigma \right) =\frac{1}{2}\left\Vert \rho -\sigma \right\Vert
_{1}$ using $F\left( \rho ,\sigma \right) $ is one of notable feature of
this quantity. It turns out that $F_{\min }\left( \rho ,\sigma \right) $
gives analogous bounds of $\Delta _{\max }\left( \rho ,\sigma \right) $,
which is another quantum version of statistical distance $\Delta \left(
p,q\right) =\frac{1}{2}\left\Vert p-q\right\Vert _{1}$, defined using
reverse test.

\section{Classical fidelity and fidelity}

We consider probability distributions over finite set $\mathcal{X}$ with $%
\left\vert \mathcal{X}\right\vert =d^{\prime }$, and quantum states $%
\mathcal{S}\left( \mathcal{H}\right) $ over $d$-dimensional Hilbert space $%
\mathcal{H}$. Define classical and quantum fidelity $F$ by

\begin{eqnarray*}
F\left( p,q\right) &:&=\sum_{x\in \mathcal{X}}\sqrt{p\left( x\right) }\sqrt{%
q\left( x\right) }, \\
F\left( \rho ,\sigma \right) &:&=\max_{U\text{:unitary}}\mathrm{tr}\,\sqrt{%
\rho }\sqrt{\sigma }U \\
&=&\mathrm{tr}\,\sqrt{\sqrt{\sigma }\rho \sqrt{\sigma }}.
\end{eqnarray*}%
Known facts about them are :

\begin{itemize}
\item 
\begin{equation}
F\left( \rho ,\sigma \right) =\min_{M:\text{measurement}}F\left( M\left(
\rho \right) ,M\left( \sigma \right) \right)  \label{F=maxFM}
\end{equation}%
where $M\left( \rho \right) $ is the probability distribution of measurement 
$M$ applied to $\rho $.

\item 
\begin{equation}
F\left( \rho ,\sigma \right) =\mathrm{tr}\,W_{\rho }^{\dagger }W_{\sigma },
\label{F=trWW}
\end{equation}%
where $W_{\rho }$ and $W_{\sigma }$ are $d\times d$ matrices with 
\begin{eqnarray*}
W_{\rho }W_{\rho }^{\dagger } &=&\rho ,\,W_{\sigma }W_{\sigma }^{\dagger
}=\sigma , \\
W_{\rho }^{\dagger }W_{\sigma } &=&W_{\sigma }^{\dagger }W_{\rho }\geq 0.
\end{eqnarray*}

\item Given parameterized family $\left\{ p_{t}\right\} $ and $\left\{ \rho
_{t}\right\} $, define \textit{Fisher information} $J_{t}$ and \textit{%
symmetric logarithmic derivative (SLD) Fisher information} $J_{t}^{S}$ by 
\begin{eqnarray*}
J_{t} &:&=\sum_{x\in \mathcal{X}}\left( l_{t}\left( x\right) \right)
^{2}p_{t}\left( x\right) , \\
J_{t}^{S} &:&=\mathrm{tr}\,\left( L_{t}^{S}\right) ^{2}\rho _{t},
\end{eqnarray*}%
where $l_{t}$, called \textit{logarithmic derivative}, is defined by $%
l_{t}\left( x\right) :=\frac{\mathrm{d}\,}{\mathrm{d}t}\log p_{t}\left(
x\right) $, and $L_{t}^{S}$, called \textit{symmetric logarithmic derivative
(SLD)}, is a solution to a linear equation \ 
\begin{equation*}
\frac{\mathrm{d}\rho _{t}}{\mathrm{d}t}=\frac{1}{2}\left\{ L_{t}^{S}\rho
_{t}+\rho _{t}\,L_{t}^{S}\right\} .
\end{equation*}%
Then,%
\begin{eqnarray*}
F\left( p_{t},p_{t+\varepsilon }\right) &=&1-\frac{1}{8}J_{t}\varepsilon
^{2}+o\left( \varepsilon ^{2}\right) , \\
F\left( \rho _{t},\rho _{t+\varepsilon }\right) &=&1-\frac{1}{8}%
J_{t}^{S}\varepsilon ^{2}+o\left( \varepsilon ^{2}\right) ,
\end{eqnarray*}%
and 
\begin{eqnarray}
F\left( p,q\right) &=&\cos \min_{C}\frac{1}{2}\int_{C}\sqrt{J_{t}}dt,
\label{F=int-J} \\
F\left( \rho ,\sigma \right) &=&\cos \min_{C}\frac{1}{2}\int_{C}\sqrt{%
J_{t}^{S}}dt,  \label{F=int-JS}
\end{eqnarray}%
where minimum is taken for all the paths with $p_{0}=p$, $p_{1}=q$, and \ $%
\rho _{0}=\rho $, $\rho _{1}=\sigma $, respectively.

\item (symmetry) 
\begin{equation*}
F\left( \rho ,\sigma \right) =F\left( \sigma ,\rho \right) .
\end{equation*}

\item (Monotonicity)%
\begin{equation}
F\left( \rho ,\sigma \right) \leq F\left( \Lambda \left( \rho \right)
,\Lambda \left( \sigma \right) \right) .  \label{monotone}
\end{equation}

\item (triangle inequality)%
\begin{eqnarray}
\cos ^{-1}F\left( \rho ,\sigma \right) &\leq &\cos ^{-1}F\left( \rho ,\tau
\right) +\cos ^{-1}F\left( \tau ,\sigma \right) ,  \label{triangle-1} \\
F\left( \rho ,\sigma \right) &\geq &2F\left( \rho ,\tau \right) F\left( \tau
,\sigma \right) -1.  \label{triangle-2}
\end{eqnarray}

\item (strong joint concavity)%
\begin{equation}
F\left( \sum_{y\in \mathcal{Y}}\lambda _{y}\rho _{y},\sum_{y\in \mathcal{Y}%
}\mu _{y}\sigma _{y}\right) \geq \sum_{y\in \mathcal{Y}}\sqrt{\lambda
_{y}\mu _{y}}F\left( \rho _{y},\sigma _{y}\right)   \label{strong-concave}
\end{equation}

\item (Multiplicativity) 
\begin{equation*}
F\left( \rho ^{\otimes n},\sigma ^{\otimes n}\right) =F\left( \rho ,\sigma
\right) ^{n}
\end{equation*}
\end{itemize}

In the paper, we consider quantities satisfying:

\begin{description}
\item[(N)] $F^{Q}\left( p,q\right) =F\left( p,q\right) $, for all the
probability distributions $p$, $q$.

\item[(M)] $F^{Q}\left( \rho ,\sigma \right) \leq F^{Q}\left( \Lambda \left(
\rho \right) ,\Lambda \left( \sigma \right) \right) $
\end{description}

\section{Another quantum analogue of classical fidelity}

A triplet $\left( \Phi ,\left\{ p,q\right\} \right) $ of a CPTP map $\Phi $
and probability distributions $p$, $q$ over the set $\mathcal{X}$ ($%
\left\vert \mathcal{X}\right\vert =d^{\prime }$) with 
\begin{equation}
\Phi \left( p\right) =\rho ,\Phi \left( q\right) =\sigma ,
\label{reverse-test}
\end{equation}%
is called \textit{reverse test} of $\left\{ \rho ,\sigma \right\} $. A
reverse test $\left( \Phi ,\left\{ p,q\right\} \right) $ with (\ref%
{reverse-test}) is said to be \textit{minimal} satisfying when $\left\vert 
\mathcal{X}\right\vert =d=\dim \mathcal{H}$. Define 
\begin{equation}
F_{\min }\left( \rho ,\sigma \right) :=\max_{\left( \Phi ,\left\{
p,q\right\} \right) \text{:(\ref{reverse-test})}}F\left( p,q\right) .
\label{def:fmin}
\end{equation}

\begin{theorem}
\label{th:fmin}Suppose $F^{Q}\left( p,q\right) $ (N) and (M). Then, $F_{\min
}\left( \rho ,\sigma \right) \leq F^{Q}\left( \rho ,\sigma \right) \leq
F\left( \rho ,\sigma \right) $. Also, $F_{\min }\left( \rho ,\sigma \right) $
satisfies (N) and (M).
\end{theorem}

\begin{proof}
Let $M$ be a measurement achieving the minimum of (\ref{F=maxFM}). Then, 
\begin{eqnarray*}
&&F^{Q}\left( \rho ,\sigma \right) \underset{\text{(M)}}{\leq }F^{Q}\left(
M\left( \rho \right) ,M\left( \sigma \right) \right) \underset{\text{(N)}}{=}%
F\left( M\left( \rho \right) ,M\left( \sigma \right) \right)  \\
&=&F\left( \rho ,\sigma \right) .
\end{eqnarray*}%
Let $\Phi $ be a CPTP map achieving the maximum of (\ref{def:fmin}). Then%
\begin{eqnarray*}
&&F^{Q}\left( \rho ,\sigma \right) =F^{Q}\left( \Phi \left( p\right) ,\Phi
\left( q\right) \right) \underset{\text{(M)}}{\geq }F^{Q}\left( p,q\right) 
\underset{\text{(N)}}{=}F\left( p,q\right)  \\
&=&F_{\min }\left( \rho ,\sigma \right) .
\end{eqnarray*}%
Obviously, $F_{\min }\left( p,q\right) \geq F\left( p,q\right) $. Also, for
any $p^{\prime }$, $q^{\prime }$ with $p=\Phi \left( p^{\prime }\right) $, $%
q=\Phi \left( q^{\prime }\right) $, $F\left( p^{\prime },q^{\prime }\right)
\leq F\left( p,q\right) $, by (\ref{monotone}). Therefore, $F_{\min }\left(
p,q\right) \leq F\left( p,q\right) $, and we have (N).

That $F_{\min }$ satisfies (M) is proved as follows.%
\begin{eqnarray*}
F_{\min }\left( \Lambda \left( \rho \right) ,\Lambda \left( \sigma \right)
\right) &=&\max \left\{ F\left( p,q\right) \,;\Phi \text{: CPTP, }\Phi
\left( p\right) =\Lambda \left( \rho \right) ,\Phi \left( q\right) =\Lambda
\left( \sigma \right) \right\} \\
&\geq &\max \left\{ F\left( p,q\right) \,;\text{ }\Phi =\Phi ^{\prime }\circ
\Lambda \text{, }\Phi ^{\prime }\left( p\right) =\rho ,\Phi ^{\prime }\left(
q\right) =\sigma \right\} \\
&=&\max \left\{ F\left( p,q\right) \,;\text{ }\Phi ^{\prime }\left( p\right)
=\rho ,\Phi ^{\prime }\left( q\right) =\sigma \right\} \\
&=&F_{\min }\left( \rho ,\sigma \right) .
\end{eqnarray*}
\end{proof}

\begin{theorem}
\label{th:fmin=fr}Suppose $\rho $ and $\sigma $ are strictly positive. Then, 
\begin{equation*}
F_{\min }\left( \rho ,\sigma \right) =\mathrm{tr}\,\rho \sqrt{\rho
^{-1/2}\sigma \rho ^{-1/2}},
\end{equation*}%
and the maximum of (\ref{def:fmin}) is achieved by any minimal reverse test $%
\left( \Phi ,\left\{ p,q\right\} \right) $.
\end{theorem}

The proof will be given later.

\begin{remark}
Using geometric mean $A\#B=\sqrt{A}\sqrt{A^{-1/2}BA^{-1/2}}\sqrt{A}$%
\thinspace \cite{Ando}, 
\begin{equation}
F_{\min }\left( \rho ,\sigma \right) =\mathrm{tr}\,\left( \rho \#\sigma
\right) .  \label{F=mean}
\end{equation}%
Hence,  the well-known property of $\#$ 
\begin{equation*}
\Lambda \left( \rho \right) \#\Lambda \left( \sigma \right) \geq \Lambda
\left( \rho \#\sigma \right) 
\end{equation*}
immediately implies that  $F_{\min }$ satisfies (M).
\end{remark}

\begin{remark}
In \cite{Matsumoto-2}\cite{Matsumoto-3}, minimization of $\mathrm{D}\left(
p||q\right) :=\sum_{x\in \mathcal{X}}p\left( x\right) \ln \frac{p\left(
x\right) }{q\left( x\right) }$ over all the reverse tests $\left( \Phi
,\left\{ p,q\right\} \right) $ is considered, and it is shown that the
minimum is achieved also by any minimal reverse test.
\end{remark}

\section{Listing all reverse tests}

\label{sec:list-of-reverse}

In this section, to solve maximization (\ref{def:fmin}), we give full
characterization of all reverse tests $\left( \Phi ,\left\{ p,q\right\}
\right) $ of $\left\{ \rho ,\sigma \right\} $ with \ $\Phi \left( \delta
_{x}\right) $ being a pure state, where $\delta _{x_{0}}\left( x\right) $
denotes a probability distribution concentrated at $x=x_{0}$ ($x,x_{0}\in 
\mathcal{X}$). 

The reason for such a restriction to be made is as follows. If 
\begin{equation*}
\Phi \left( \delta _{x}\right) =\rho _{x}=\sum_{y\in \mathcal{Y}}s\left(
y|x\right) \left\vert \varphi _{xy}\right\rangle \left\langle \varphi
_{xy}\right\vert ,
\end{equation*}%
let $\Phi ^{\prime }$ a CPTP map from probability distributions over $%
\mathcal{X\times Y}$ to $\mathcal{S}\left( \mathcal{H}\right) $ such that 
\begin{equation*}
\Phi ^{\prime }\left( \delta _{\left( x,y\right) }\right) =\left\vert
\varphi _{xy}\right\rangle \left\langle \varphi _{xy}\right\vert .
\end{equation*}%
Then \ if $\Phi \left( p\right) =\rho $ and $\Phi \left( q\right) =\sigma $, 
$\Phi ^{\prime }\left( p^{\prime }\right) =\rho $ and $\Phi ^{\prime }\left(
q^{\prime }\right) =\sigma $, where 
\begin{equation*}
p^{\prime }\left( x,y\right) =p\left( x\right) s\left( y|x\right)
,\,\,q^{\prime }\left( x,y\right) =q\left( y\right) s\left( y|x\right) .
\end{equation*}%
Hence, $\left( \Phi ^{\prime },\left\{ p^{\prime },q^{\prime }\right\}
\right) $ is a reverse test of $\left\{ \rho ,\sigma \right\} $ with $\Phi
\left( \delta _{x}\right) $ being a pure state, and $p=\Psi _{1}\left(
p^{\prime }\right) $, $q=\Psi _{1}\left( q^{\prime }\right) $, where $\Psi
_{1}$ is taking marginal over $\mathcal{Y}$. Hence, $F\left( p.q\right) \geq
F\left( p^{\prime },q^{\prime }\right) $. Also, observe $p^{\prime }=\Psi
_{2}\left( p\right) $, $q^{\prime }=\Psi _{2}\left( q\right) $, where      
\begin{equation*}
\Psi _{2}:r\left( x\right) \rightarrow r\left( x\right) s\left( y|x\right) .
\end{equation*}%
Therefore, after all,   
\begin{equation*}
F\left( p.q\right) =F\left( p^{\prime },q^{\prime }\right) .
\end{equation*}%
Therefore, we can replace $\Phi $ by $\Phi ^{\prime }$.

Below, we \ indicate $\Phi $ by $\ $a matrix $N$, whose $x$th column vector
is $\left\vert \varphi _{x}\right\rangle $ with $\left\vert \varphi
_{x}\right\rangle \left\langle \varphi _{x}\right\vert =\Phi \left( \delta
_{x}\right) $. Then the condition (\ref{reverse-test}) is rewritten as 
\begin{equation}
\sum_{x\in \mathcal{X}}p\left( x\right) \left\vert \varphi _{x}\right\rangle
\left\langle \varphi _{x}\right\vert =\rho ,\sum_{x\in \mathcal{X}}q\left(
x\right) \left\vert \varphi _{x}\right\rangle \left\langle \varphi
_{x}\right\vert =\sigma   \label{sum=rho}
\end{equation}%
Here note in general, $d^{\prime }$ can be larger than $d=\dim \mathcal{H}$.

\begin{lemma}
\label{lem:ww=rho}$WW^{\dagger }=\rho $ if and only if 
\begin{equation*}
W=\sqrt{\rho }U,
\end{equation*}%
with $U$ being isometry, $UU^{\dagger }=\mathbf{1}$.
\end{lemma}

\begin{proof}
We only have to show `only if'. Suppose $WW^{\dagger }=\rho $. Then, letting 
$\sqrt{\rho }^{-1}$ be the (Moore-Penrose) generalized inverse of $\sqrt{%
\rho }$, we have 
\begin{equation*}
\left( \sqrt{\rho }^{-1}W\right) \left( \sqrt{\rho }^{-1}W\right) ^{\dagger
}=P,
\end{equation*}%
where $P$ is the projector on the support of $\rho $. Observe 
\begin{equation*}
\left( 1-P\right) W\left( \left( 1-P\right) W\right) ^{\dagger }=\left(
1-P\right) \rho \left( 1-P\right) =0
\end{equation*}%
implies 
\begin{equation*}
\left( 1-P\right) W=0
\end{equation*}%
Therefore, 
\begin{equation*}
\sqrt{\rho }^{-1}\left( \sqrt{\rho }^{-1}W\right) =PW=W.
\end{equation*}%
Let $U^{\prime }$ be a partial isometry from $\ker \,\left( \sqrt{\rho }%
^{-1}W\right) ^{\dagger }$ to $\ker \rho $, 
\begin{equation*}
U=\sqrt{\rho }^{-1}W\oplus U^{\prime }
\end{equation*}%
satisfies the requirement. Hence, we have the assertion.
\end{proof}

In the reminder of the section, we suppose $\rho >0$. Let $D_{p}$ and $D_{q}$
be 
\begin{equation}
D_{p}=\sqrt{\mathrm{diag}\,\left( p_{1},\cdots ,p_{d^{\prime }}\right) }%
,\,\,D_{q}=\sqrt{\mathrm{diag}\,\left( q_{1},\cdots ,q_{d^{\prime }}\right) }%
.  \label{DD=p}
\end{equation}%
Then, by (\ref{sum=rho}), $ND_{p}\left( ND_{p}\right) ^{\dagger }=\rho $ and 
$ND_{q}\left( ND_{q}\right) ^{\dagger }=\sigma $. Also,  
\begin{equation}
T:=\sqrt{\rho ^{-1/2}\sigma \rho ^{-1/2}}\,\,\left( >0\right)   \label{Tdef}
\end{equation}%
satisfies $\left( \sqrt{\rho }T\right) \left( \sqrt{\rho }T\right) ^{\dagger
}=\sigma $. \ Therefore, 
\begin{equation}
\left[ \sqrt{\rho }\,\,0\right] V=ND_{p},\,\left[ \sqrt{\rho }T\,\,0\right]
U=ND_{p},  \label{rhov=ND}
\end{equation}%
for some $U$,$V\in \mathrm{U}\left( \mathcal{H}^{\prime }\right) $, $\dim 
\mathcal{H}^{\prime }=d^{\prime }$. Then, 
\begin{eqnarray}
ND_{q}D_{p}N^{\dagger } &=&\left[ \sqrt{\rho }\,T\,0\right] UV^{\dagger }%
\left[ 
\begin{array}{c}
\sqrt{\rho }\, \\ 
0%
\end{array}%
\right]   \notag \\
&=&\sqrt{\rho }TA\sqrt{\rho }\geq 0,  \label{NDN=rhoTArho}
\end{eqnarray}%
where, with $P$ being the projector onto $\mathcal{H}$, $A=PUV^{\dagger }P$.
Therefore,%
\begin{equation}
TA=A^{\dagger }T\geq 0.  \label{AT=TA}
\end{equation}
Also, by (\ref{rhov=ND}), 
\begin{eqnarray*}
\left[ \sqrt{\rho }\,\,0\right] V\left( D_{p}^{-1}D_{q}\right) V^{\dagger }
&=&\left[ \sqrt{\rho }T\,\,0\right] UV^{\dagger } \\
&=&\left[ \sqrt{\rho }\,\,0\right] \left[ 
\begin{array}{cc}
TA & TA^{\prime } \\ 
A^{\prime \dagger }T & C%
\end{array}%
\right] ,
\end{eqnarray*}%
where $A^{\prime }:=PU\left( 1-P\right) $ and $C>0$. If $\dot{C}>0$ satisfy 
\begin{equation}
\ker TA^{\prime }\subset \ker C,\,\,\left( TA\right) ^{-1/2}\left(
TA^{\prime }\right) C^{-1/2}\leq \mathbf{1,}  \label{C>}
\end{equation}%
we have%
\begin{equation}
\tilde{T}:=\left[ 
\begin{array}{cc}
TA & TA^{\prime } \\ 
A^{\prime \dagger }T & C%
\end{array}%
\right] \geq 0.  \label{L}
\end{equation}

Note $\left\{ A;\left\Vert A\right\Vert \leq 1\right\} $ is identical to the
totality of matrices with the form $A=PUP$, where $U\in \mathrm{U}\left( 
\mathcal{H}^{\prime }\right) $, $\mathcal{H}\subset \mathcal{H}^{\prime }$, $%
\dim \mathcal{H}^{\prime }=d^{\prime }$ and $P$ is the projector onto $%
\mathcal{H}$ . Then, a reverse test can be composed as indecated in the
following (i)-(v):

\begin{description}
\item[(i)]  Choose  $A\in \left\{ A;\left\Vert A\right\Vert \leq 1\right\} $
with (\ref{AT=TA}).

\item[(ii)]  Compose $A^{\prime }$ such that $[A\,A^{\prime }]$ is isometry.

\item[(iii)]  Let   $\tilde{T}$ as of  (\ref{L}),  $V$ is diagonalize it :$%
\tilde{T}=VDV^{\dagger }$. 

\item[(vi)]  Define  $\left\vert \varphi _{x}\right\rangle $ ($x\in \mathcal{%
X}$) as the normalized column vectors of $[\sqrt{\rho }\,0]V$. 
\end{description}

Finally, $p\left( x\right) $ and $q\left( x\right) $ is the square of the
magnitude of the $x$th column vector of $[\sqrt{\rho }\,0]V$ and $\,\left[ 
\sqrt{\rho }TA\,\sqrt{\rho }TA^{\prime }\right] $, respectively. $p$ and $q$
are obtained also as follows. Define $\rho ^{\prime }$,$\sigma ^{\prime }\in 
\mathcal{S}\left( \mathcal{H}^{\prime }\right) $ by 
\begin{equation*}
\rho ^{\prime }=\left[ 
\begin{array}{cc}
\rho  & 0 \\ 
0 & 0%
\end{array}%
\right] ,\sigma ^{\prime }=\tilde{T}\rho \tilde{T}\text{,}
\end{equation*}%
and let the measurement $M$ be projectors onto eigenspaces of $\tilde{T}$.
Then, $p=M\left( \rho ^{\prime }\right) $ and $M=\left( \sigma ^{\prime
}\right) $. So, the last step of the composition is:

\begin{description}
\item[(v)] Let $p=$ $M\left( \rho ^{\prime }\right) $ and $q=M\left( \sigma
^{\prime }\right) $, where $M$ is the projectors onto eigenspaces of $\tilde{%
T}$. 
\end{description}

\section{ Proof of Theorem\thinspace \protect\ref{th:fmin=fr}}

\begin{proof}
\negthinspace \textbf{of Theorem\thinspace \ref{th:fmin=fr}.}\ \ By (\ref%
{NDN=rhoTArho}),
\end{proof}

\begin{eqnarray*}
\mathrm{tr}\,\sqrt{\rho }TA\sqrt{\rho } &=&\mathrm{tr}\,ND_{q}D_{p}N^{%
\dagger } \\
&=&\mathrm{tr\,}D_{q}N^{\dagger }ND_{p} \\
&=&\mathrm{tr\,}D_{q}D_{p},
\end{eqnarray*}%
where the last identity is due to $\left( N^{\dagger }N\right) _{ii}=1$. On
the other hand, (\ref{AT=TA}) implies, 
\begin{equation}
T^{2}-\left( TA\right) ^{2}=T^{2}-TAA^{\dagger }T=T\left( 1-AA^{\dagger
}\right) T\geq 0.  \label{TAT>0}
\end{equation}%
where the last inequality is due to $\left\Vert A\right\Vert \leq 1$. Since $%
\sqrt{t}$ is operator monotone, 
\begin{equation*}
TA\leq T.
\end{equation*}%
Therefore, 
\begin{eqnarray*}
F_{\min }\left( \rho ,\sigma \right)  &=&\max_{p,q:\text{(\ref{sum=rho})}%
}\sum_{x\in \mathcal{X}}\sqrt{p\left( x\right) q\left( x\right) } \\
&=&\mathrm{tr\,}D_{q}D_{p} \\
&=&\mathrm{tr}\,\sqrt{\rho }TA\sqrt{\rho } \\
&\leq &\mathrm{tr}\,\sqrt{\rho }T\sqrt{\rho } \\
&=&\mathrm{tr}\,\rho \sqrt{\rho ^{-1/2}\sigma \rho ^{-1/2}}.
\end{eqnarray*}%
The inequality is achieved when $A=\mathbf{1}$, which corresponds to minimal
reverse tests.

\section{Seeing from `behind'}

By Theorem\thinspace \ref{th:fmin=fr}\ ,,

\begin{equation*}
F_{\min }\left( \rho ,\sigma \right) =\mathrm{tr}\,W_{\rho }W_{\sigma
}^{\dagger },
\end{equation*}%
where $W_{\rho }:=\sqrt{\rho }$ and $W_{\sigma }:=\sqrt{\rho }T$, with $T$
being as of (\ref{Tdef}). Observe 
\begin{equation*}
W_{\rho }W_{\sigma }^{\dagger }=W_{\sigma }W_{\rho }^{\dagger }.
\end{equation*}%
Therefore, by (\ref{F=trWW}), we have 
\begin{equation}
F_{\min }\left( \rho ,\sigma \right) =F\left( \rho ^{\prime },\sigma
^{\prime }\right) ,  \label{F=Fmin}
\end{equation}%
where%
\begin{equation*}
\rho ^{\prime }:=W_{\rho }^{\dagger }W_{\rho }=\rho ,\,\,\,\sigma ^{\prime
}:=W_{\sigma }^{\dagger }W_{\sigma }=T\rho T.
\end{equation*}

A meaning of (\ref{F=Fmin}) is given in the sequel. Letting

\begin{equation*}
W_{\rho }=\sum_{i=1}^{d}\sum_{x\in \mathcal{X}}w_{\rho ,i\,x}\left\vert
i\right\rangle \left\langle e_{x}\right\vert ,\,\,W_{\sigma
}=\sum_{i=1}^{d}\sum_{x\in \mathcal{X}}w_{\sigma ,i\,x}\left\vert
i\right\rangle \left\langle e_{x}\right\vert ,
\end{equation*}%
define

\begin{equation*}
\left\vert W_{\rho }\right\rangle :=\sum_{i=1}^{d}\sum_{x\in \mathcal{X}%
}w_{\rho ,i\,x}\left\vert i\right\rangle \left\vert e_{x}\right\rangle
,\,\left\vert W_{\sigma }\right\rangle :=\sum_{i=1}^{d}\sum_{x\in \mathcal{X}%
}w_{\sigma ,i\,x}\left\vert i\right\rangle \left\vert e_{x}\right\rangle .
\end{equation*}%
Then, one can easily check

\begin{eqnarray*}
\rho &=&\mathrm{tr}\,_{\mathcal{H}^{\prime }}\left\vert W_{\rho
}\right\rangle \left\langle W_{\rho }\right\vert ,\,\,\,\sigma :=\mathrm{tr}%
\,_{\mathcal{H}^{\prime }}\left\vert W_{\sigma }\right\rangle \left\langle
W_{\sigma }\right\vert \\
\rho ^{\prime } &=&\mathrm{tr}\,_{\mathcal{H}}\left\vert W_{\rho
}\right\rangle \left\langle W_{\rho }\right\vert ,\,\,\,\sigma ^{\prime }=%
\mathrm{tr}\,_{\mathcal{H}}\left\vert W_{\sigma }\right\rangle \left\langle
W_{\sigma }\right\vert
\end{eqnarray*}%
hold. Hence, $F_{\min }\left( \rho ,\sigma \right) $ equals fidelity of
`hidden' part of the purification of $\rho $ and $\sigma $.

\section{$F_{\min }$ for pure states}

Any reverse test $\left( \Phi ,\left\{ p,q\right\} \right) $ of $\left\{
\rho ,\left\vert \varphi \right\rangle \right\} $ is in the following form: 
\begin{eqnarray}
\Phi \left( \delta _{x}\right)  &=&\left\vert \varphi \right\rangle
\left\langle \varphi \right\vert ,\,\,x\in \mathrm{supp}\,q\,,  \notag \\
\rho  &=&\Phi \left( p\right) =c\left\vert \varphi \right\rangle
\left\langle \varphi \right\vert +\sum_{x\notin \mathrm{supp}\,q\,}p\left(
x\right) \Phi \left( \delta _{x}\right) ,  \label{reverse-pure}
\end{eqnarray}%
where $c:=\sum_{x\in \mathrm{supp}\,q\,}p\left( x\right) $. Therefore, by
monotonicity of Fidelity by CPTP maps, 
\begin{equation*}
\sum_{x\in \mathcal{X}}\sqrt{p\left( x\right) q\left( x\right) }\leq \sqrt{%
c\cdot 1}+\sqrt{c\cdot 0}=\sqrt{c},
\end{equation*}%
and the inequality is achieved by the following $q\left( x\right) $ and $%
p\left( x\right) $:%
\begin{equation*}
q\left( x\right) =\delta _{x_{0}},\,p\left( x\right) =c\delta
_{x_{0}}\,\,(x\in \mathrm{supp}\,q),
\end{equation*}%
where $x_{0}$ is a point in $\mathrm{supp}\,q$. 

Therefore, we maximize $c$ with $\rho -c\left\vert \varphi \right\rangle
\left\langle \varphi \right\vert \geq 0$,    or equivalently, if $\left\vert
\varphi \right\rangle \in \mathrm{supp}\,\rho $, 
\begin{equation*}
\mathbf{1}-c\rho ^{-1/2}\left\vert \varphi \right\rangle \left\langle
\varphi \right\vert \rho ^{-1/2}\geq 0.
\end{equation*}%
Therefore, if $\left\vert \varphi \right\rangle \in \mathrm{supp}\,\rho $,  
\begin{eqnarray*}
F_{\min }\left( \rho ,\left\vert \varphi \right\rangle \right)  &=&\sqrt{%
\left( \mathrm{tr}\,\rho ^{-1/2}\left\vert \varphi \right\rangle
\left\langle \varphi \right\vert \rho ^{-1/2}\right) ^{-1}} \\
&=&\left\Vert \sqrt{\rho }^{-1}\left\vert \varphi \right\rangle \right\Vert
^{-1}.
\end{eqnarray*}

In case $\left\vert \varphi \right\rangle \notin \mathrm{supp}\,\,\rho $,
the maximum of $c$ with $\rho -c\left\vert \varphi \right\rangle
\left\langle \varphi \right\vert \geq 0$ is zero, and%
\begin{equation*}
F_{\min }\left( \rho ,\left\vert \varphi \right\rangle \right) =0.
\end{equation*}%
In particular, 
\begin{equation}
F_{\min }\left( \left\vert \psi \right\rangle ,\left\vert \varphi
\right\rangle \right) =0.  \label{F-pure}
\end{equation}

\section{Properties of $F_{\min }$}

\begin{proposition}
\label{prop:no-triangle}$F_{\min }\left( \rho ,\sigma \right) $ does not
satisfy triangle inequalities: there is $\rho $, $\sigma $, and $\tau $ with%
\begin{eqnarray*}
\cos ^{-1}F_{\min }\left( \rho ,\sigma \right)  &>&\cos ^{-1}F_{\min }\left(
\rho ,\tau \right) +\cos ^{-1}F_{\min }\left( \tau ,\sigma \right) , \\
F_{\min }\left( \rho ,\sigma \right)  &<&2F_{\min }\left( \rho ,\tau \right)
F_{\min }\left( \tau ,\sigma \right) -1.
\end{eqnarray*}
\end{proposition}

\begin{proof}
We let 
\begin{eqnarray*}
\rho &=&\left\vert \psi \right\rangle \left\langle \psi \right\vert ,\sigma
=\left\vert \varphi \right\rangle \left\langle \varphi \right\vert ,\, \\
\left\vert \psi \right\rangle &=&\left[ 
\begin{array}{c}
\cos \frac{\theta }{2} \\ 
\sin \frac{\theta }{2}%
\end{array}%
\right] ,\left\vert \varphi \right\rangle =\left[ 
\begin{array}{c}
\cos \frac{\theta }{2} \\ 
-\sin \frac{\theta }{2}%
\end{array}%
\right] , \\
\tau &=&\frac{1}{\left\vert \cos \frac{\theta }{2}\right\vert +\left\vert
\sin \frac{\theta }{2}\right\vert }\left[ 
\begin{array}{cc}
\left\vert \cos \frac{\theta }{2}\right\vert & 0 \\ 
0 & \left\vert \sin \frac{\theta }{2}\right\vert%
\end{array}%
\right] \,.
\end{eqnarray*}%
Then, 
\begin{eqnarray*}
F_{\min }\left( \left\vert \psi \right\rangle ,\left\vert \varphi
\right\rangle \right) &=&0, \\
F_{\min }\left( \left\vert \psi \right\rangle ,\tau \right) &=&F_{\min
}\left( \left\vert \varphi \right\rangle ,\tau \right) =\frac{1}{\left\vert
\cos \frac{\theta }{2}\right\vert +\left\vert \sin \frac{\theta }{2}%
\right\vert }.
\end{eqnarray*}%
Hence, letting $\theta $ be small enough, we have asserted inequalities. (In
fact, the inequality is satisfied all $\theta $ lying between $0$ and $\frac{%
\pi }{2}$.)
\end{proof}

\begin{proposition}
$F_{\min }\left( \rho ,\sigma \right) =F_{\min }\left( \sigma ,\rho \right) $%
, $F_{\min }\left( \rho ^{\otimes n},\sigma ^{\otimes n}\right) =F_{\min
}\left( \rho ,\sigma \right) ^{n}$. \ 
\end{proposition}

\begin{proof}
Trivial by definition.
\end{proof}

\begin{theorem}
\label{th:strong-concave}(Strong concavity) 
\begin{equation*}
F_{\min }\left( \sum_{y\in \mathcal{Y}}\lambda _{y}\rho _{y},\sum_{y\in 
\mathcal{Y}}\mu _{y}\sigma _{y}\right) \geq \sum_{y\in \mathcal{Y}}\sqrt{%
\lambda _{y}\mu _{y}}F_{\min }\left( \rho _{y},\sigma _{y}\right)
\end{equation*}%
{}
\end{theorem}

\begin{proof}
Let $\left( \Phi _{y},\left\{ p_{y},q_{y}\right\} \right) $ be a reverse
test of $\left\{ \rho _{y},\sigma _{y}\right\} $ with 
\begin{equation*}
F\left( p_{y},q_{y}\right) =F_{\min }\left( \rho _{y},\sigma _{y}\right) ,
\end{equation*}%
where $p_{y}$, $q_{y}$ are probability distributions over $\mathcal{X}$ with 
$\left\vert \mathcal{X}\right\vert =d^{\prime }=d$ (minimal). Define $\tilde{%
p}_{y_{0}}\left( x,y\right) :=p_{y_{0}}\left( x\right) \delta _{y_{0}}\left(
y\right) $, $\tilde{q}_{y_{0}}\left( x,y\right) =q_{y_{0}}\left( x\right)
\delta _{y_{0}}\left( y\right) $, and $\tilde{\Phi}\left( \delta _{\left(
x,y\right) }\right) :=\Phi _{y}\left( \delta _{x}\right) $. Then, 
\begin{eqnarray*}
\tilde{\Phi}\left( \tilde{p}_{y_{0}}\right)  &=&\sum_{x\in \mathcal{X},y\in 
\mathcal{Y}}p_{y_{0}}\left( x\right) \delta _{y_{0}}\left( y\right) \Lambda
_{y}\left( \delta _{x}\right) =\rho _{y_{0}}, \\
\tilde{\Phi}\left( \tilde{q}_{y_{0}}\right)  &=&\sigma _{y_{0}}, \\
F\left( \tilde{p}_{y_{0}},\tilde{q}_{y_{0}}\right)  &=&\sum_{x\in \mathcal{X}%
,y\in \mathcal{Y}}\sqrt{p_{y_{0}}\left( x\right) q_{y_{0}}\left( x\right) }%
\delta _{y_{0}}\left( y\right)  \\
&=&F\left( p_{y_{0}},q_{y_{0}}\right) =F_{\min }\left( \rho _{y_{0}},\sigma
_{y_{0}}\right) .
\end{eqnarray*}%
Therefore, 
\begin{eqnarray*}
&&\sum_{y\in \mathcal{Y}}\sqrt{\lambda _{y}\mu _{y}}F_{\min }\left( \rho
_{y},\sigma _{y}\right) =\sum_{y\in \mathcal{Y}}\sqrt{\lambda _{y}\mu _{y}}%
F\left( \tilde{p}_{y},\tilde{q}_{y}\right)  \\
&\leq &F\left( \sum_{y\in \mathcal{Y}}\lambda _{y}\tilde{p}_{y},\sum_{y\in 
\mathcal{Y}}\mu _{y}\tilde{q}_{y}\right) \leq F_{\min }\left( \tilde{\Phi}%
\left( \sum_{y\in \mathcal{Y}}\lambda _{y}\tilde{p}_{y}\right) ,\tilde{\Phi}%
\left( \sum_{y\in \mathcal{Y}}\mu _{y}\tilde{q}_{y}\right) \right)  \\
&=&F_{\min }\left( \sum_{y\in \mathcal{Y}}\lambda _{y}\rho _{y},\sum_{y\in 
\mathcal{Y}}\mu _{y}\sigma _{y}\right) ,
\end{eqnarray*}%
where the second line is due to (\ref{strong-concave}) for classical
fidelity, and the third and the fourth line is due to (N) and (M),
respectively.
\end{proof}

\begin{remark}
Alternative proof is given using (\ref{F=mean}) and the following property
of $\#\,$\cite{Ando}: 
\begin{equation*}
A\#C+B\#D\leq \left( A+B\right) \#\left( C+D\right) .
\end{equation*}
\end{remark}

\section{Generalization to $F_{f}$}

As noted before, both of the minimum of $\mathrm{D}\left( p||q\right) $ \
and the maximum of $F\left( p,q\right) $ are achieved by minimal ones. This
section tries generalization of 

Let $f$ be operator monotone function on $[0,\infty )$. Then one can define%
\begin{equation*}
F_{f}\left( p,q\right) :=\sum_{x\in \mathcal{X}}p\left( x\right) \,f\left( 
\frac{q\left( x\right) }{p\left( x\right) }\right) .
\end{equation*}%
An example is $f\left( t\right) =t^{\alpha }$ ($0<\alpha <1$), 
\begin{equation*}
F_{\alpha }\left( p,q\right) =\sum_{x\in \mathcal{X}}p^{1-\alpha }\left(
x\right) q^{\alpha }\left( x\right) .
\end{equation*}%
Note 
\begin{equation*}
\mathrm{D}\left( p||q\right) =\lim_{\alpha \rightarrow 1}\frac{-1}{1-\alpha }%
\ln F_{1-\alpha }\left( p,q\right) .
\end{equation*}%
Thus, this family interpolates between $\mathrm{D}\left( p||q\right) $ and $%
F\left( p,q\right) $. 

Their quantum analogue is defined as follows. 
\begin{equation*}
F_{f}^{\min }\left( \rho ,\sigma \right) =\max F_{f}\left( p,q\right) 
\end{equation*}%
here the maximum is taken over all the reverse tests $\left( \Phi ,\left\{
p,q\right\} \right) $ of $\left\{ \rho ,\sigma \right\} $. 

\begin{theorem}
\label{th:FQf>Ff}If $F_{f}^{Q}\left( \rho ,\sigma \right) $ satisfies $%
F_{f}^{Q}\left( p,q\right) =F_{f}\left( p,q\right) $ for all probability
distributions $p$,$q$, and monotone increasing by any CPTP map, 
\begin{equation*}
F_{f}^{\min }\left( \rho ,\sigma \right) \leq F_{f}^{Q}\left( \rho ,\sigma
\right) .
\end{equation*}
\end{theorem}

\begin{proof}
Almost parallel with the proof of Theorem\thinspace \ref{th:fmin}, thus
omitted.
\end{proof}

\begin{theorem}
\begin{equation*}
F_{f}^{\min }\left( \lambda \rho _{0}+\left( 1-\lambda \right) \rho
_{1},\lambda \sigma _{0}+\left( 1-\lambda \right) \sigma _{1}\right) \geq
\lambda F_{f}^{\min }\left( \rho _{0},\sigma _{0}\right) +\left( 1-\lambda
\right) F_{f}^{\min }\left( \rho _{1},\sigma _{1}\right) .
\end{equation*}
\end{theorem}

\begin{proof}
Almost parallel with the proof of Theorem\thinspace \ref{th:strong-concave},
thus omitted.
\end{proof}

Also, we define 
\begin{equation*}
F_{f}^{\prime }\left( \rho ,\sigma \right) :=\mathrm{tr}\,\rho ^{\frac{1}{2}%
}\,f\left( T^{2}\right) \rho ^{\frac{1}{2}},
\end{equation*}%
where $T$ is as of \ref{Tdef}. 

Due to the main result (the equation (3.8) ) of \cite{Ando}, one find an
operator connection $\natural $ with%
\begin{equation*}
\rho ^{\frac{1}{2}}\,f\left( T^{2}\right) \rho ^{\frac{1}{2}}=\rho
\,\natural \,\sigma \,\,\,\left( \rho >0\right) ,\,
\end{equation*}%
thus 
\begin{equation*}
F_{f}^{\prime }\left( \rho ,\sigma \right) =\mathrm{tr}\,\rho \,\natural
\,\sigma \,\,\,\,\left( \rho >0\right) .
\end{equation*}%
Any operator connection $\natural $ satisfies 
\begin{eqnarray*}
S\left( A\natural B\right) S^{\dagger } &\leq &\left( SAS^{\dagger }\right)
\natural \left( SBS^{\dagger }\right) , \\
A\natural B+C\natural D &\leq &\left( A+C\right) \natural \left( B+D\right) ,
\end{eqnarray*}%
(see Theorem\thinspace 3.5 of \cite{Ando}) which implies 
\begin{equation*}
\Lambda \left( A\natural B\right) \leq \Lambda \left( A\right) \natural
\Lambda \left( B\right) ,
\end{equation*}%
for any TPCP map $\Lambda $. This implies 
\begin{equation*}
F_{f}^{\prime }\left( \Lambda \left( \rho \right) ,\Lambda \left( \sigma
\right) \right) \geq F_{f}^{\prime }\left( \rho ,\sigma \right) .
\end{equation*}%
Since $F_{f}^{\prime }\left( p,q\right) =F\left( p,q\right) $, by
Theorem\thinspace \ref{th:FQf>Ff}, 
\begin{equation}
F_{f}^{\min }\left( \rho ,\sigma \right) \leq F_{f}^{\prime }\left( \rho
,\sigma \right) .  \label{F<F'}
\end{equation}

\begin{theorem}
Suppose $\rho >0$. Then, 
\begin{equation*}
F_{f}^{\min }\left( \rho ,\sigma \right) =F_{f}^{\prime }\left( \rho ,\sigma
\right) .
\end{equation*}%
Moreover, this number is achieved by any minimal reverse test.
\end{theorem}

\begin{proof}
By (\ref{F<F'}), we only have to show `$\geq $'. Let $\left( \Phi ,\left\{
p,q\right\} \right) $ be a minimal reverse test. Then, by the argument in
Section\thinspace \ref{sec:list-of-reverse}, $p=M\left( \rho \right) $ and $%
q=M\left( T\rho T\right) $, where $M$ is the projectors onto the
eigenvectors $\left\{ \left\vert e_{x}\right\rangle \right\} $ of $T=$ $%
T=\sum_{x\in \mathcal{X}}\lambda _{x}\left\vert e_{x}\right\rangle $ $%
\left\langle e_{x}\right\vert $, 
\begin{equation*}
p\left( x\right) =\left\langle e_{x}\right\vert \rho \left\vert
e_{x}\right\rangle ,\,\,q\left( x\right) =\lambda _{x}^{2}\left\langle
e_{x}\right\vert \rho \left\vert e_{x}\right\rangle .
\end{equation*}%
Therefore, 
\begin{eqnarray*}
F_{f}^{\min }\left( \rho ,\sigma \right)  &=&\sum_{x\in \mathcal{X}%
}\left\langle e_{x}\right\vert \rho \left\vert e_{x}\right\rangle f\left(
\lambda _{x}^{2}\right)  \\
&=&\sum_{x\in \mathcal{X}}\left\langle e_{x}\right\vert \rho f\left(
T^{2}\right) \left\vert e_{x}\right\rangle =F_{f}^{\prime }\left( \rho
,\sigma \right) .
\end{eqnarray*}%
Therefore, we have $F_{f}\left( \rho ,\sigma \right) \geq F_{f}^{\prime
}\left( \rho ,\sigma \right) $, and the proof is complete.
\end{proof}

\section{RLD Fisher information and tangent reverse estimation}

Given a parameterized family $\left\{ \rho _{t}\right\} $, we define right
logarithmic derivative (\textit{RLD) Fisher information} $J_{t}^{R}$\ by%
\begin{equation*}
J_{t}^{R}:=\mathrm{tr}\,\left( L_{t}^{R}\right) ^{\dagger }L_{t}^{R}\rho
_{t},
\end{equation*}%
where $L_{t}^{R}$ is called right logarithmic derivative and is the unique
solution to the linear equation 
\begin{equation*}
\frac{\mathrm{d}\rho _{t}}{\mathrm{d\,}t}=L_{t}^{R}\rho _{t}.
\end{equation*}%
($L_{t}^{R}$ exists if and only if $\mathrm{supp}\,\frac{\mathrm{d}\rho _{t}%
}{\mathrm{d\,}t}\subset \mathrm{supp}\,\rho _{t}$.)

A triplet $\left( \Lambda ,p_{t},\frac{\mathrm{d}p_{t}}{\mathrm{d\,}t}%
\right) $ is said to be \textit{tangent reverse estimation} of $\rho _{t}$
at $t$ if it satisfies

\begin{equation*}
\Lambda \left( p_{t}\right) =\rho _{t}\,,\,\Lambda \left( \frac{\mathrm{d}%
p_{t}}{\mathrm{d\,}t}\right) =\frac{\mathrm{d}\rho _{t}}{\mathrm{d\,}t}.
\end{equation*}%
With $\Lambda \left( \delta _{x}\right) =\left\vert \varphi
_{x}\right\rangle \left\langle \varphi _{x}\right\vert $, we have \ 
\begin{eqnarray*}
\sum_{x\in \mathcal{X}^{\prime }}p_{t}\left( x\right) \left\vert \varphi
_{x}\right\rangle \left\langle \varphi _{x}\right\vert  &=&NP_{t}N^{\dagger
}=\rho _{t},\, \\
\sum_{x\in \mathcal{X}^{\prime }}\frac{\mathrm{d}p_{t}\left( x\right) }{%
\mathrm{d\,}t}\left\vert \varphi _{x}\right\rangle \left\langle \varphi
_{x}\right\vert  &=&N\frac{\mathrm{d}P_{t}}{\mathrm{d\,}t}N^{\dagger }=\frac{%
\mathrm{d}\rho _{t}}{\mathrm{d\,}t},\,
\end{eqnarray*}%
where $N=[\left\vert \varphi _{1}\right\rangle ,\cdots ,\left\vert \varphi
_{d^{\prime }}\right\rangle ]$ and $P_{t}=\mathrm{diag}\,\left( p_{t}\left(
1\right) ,\cdots ,p\left( d^{\prime }\right) \right) $. When $d^{\prime
}=d=\dim \mathcal{H}$, we say the tangent reverse estimation is \textit{%
minimal}. It is known that 
\begin{equation*}
J_{t}^{R}=\min J_{p_{t}},
\end{equation*}%
where the minimum is taken for all the tangent reverse estimation of $\rho
_{t}$ at $t$. It is also known that the minimum is achieved by any minimal
tangent reverse estimation, and RLD satisfies%
\begin{equation}
L_{t}^{R}=NL_{t}N^{-1}  \label{rld-diagonal}
\end{equation}%
with $L_{t}=\mathrm{diag}\,\left( l_{t}\left( 1\right) ,\cdots ,l_{t}\left(
d\right) \right) $ and $l_{t}=\frac{\mathrm{d}p_{t}/\mathrm{d\,}t}{p_{t}}$.

\section{$F_{\min }$, RLD Fisher information and the shortest distance}

Consider a smooth curve $\left\{ \rho _{t}\right\} $. Suppose $\left\{ \rho
_{t}\right\} $ is lying interior of $\mathcal{S}\left( \mathcal{H}\right) $
we have

\begin{eqnarray*}
&&F_{\min }\left( \rho _{t},\rho _{t+\varepsilon }\right)  \\
&=&\mathrm{tr}\,\rho _{t}\sqrt{\rho _{t}^{-1/2}\rho _{t+\varepsilon }\rho
_{t}^{-1/2}} \\
&=&\mathrm{tr}\,\rho _{t}\sqrt{\mathbf{1}+\varepsilon \rho _{t}^{-1/2}\frac{%
\mathrm{d}\rho _{t}}{\mathrm{d\,}t}\rho _{t}^{-1/2}+\frac{\varepsilon ^{2}}{2%
}\rho _{t}^{-1/2}\frac{\mathrm{d}^{2}\rho _{t}}{\mathrm{d\,}t^{2}}\rho
_{t}^{-1/2}+O\left( \varepsilon ^{3}\right) } \\
&=&1+\mathrm{tr}\,\rho _{t}\left\{ \frac{\varepsilon }{2}\rho _{t}^{-1/2}%
\frac{\mathrm{d}\rho _{t}}{\mathrm{d\,}t}\rho _{t}^{-1/2}+\frac{\varepsilon
^{2}}{4}\rho _{t}^{-1/2}\frac{\mathrm{d}^{2}\rho _{t}}{\mathrm{d\,}t^{2}}%
\rho _{t}^{-1/2}-\frac{1}{8}\left( \varepsilon \rho _{t}^{-1/2}\frac{\mathrm{%
d}\rho _{t}}{\mathrm{d\,}t}\rho _{t}^{-1/2}\right) ^{2}\right\} +O\left(
\varepsilon ^{3}\right)  \\
&=&1-\frac{\varepsilon ^{2}}{8}J_{t}^{R}+O\left( \varepsilon ^{3}\right) .
\end{eqnarray*}%
To evaluate the $O\left( \varepsilon ^{3}\right) $-term above, we compute $%
\frac{\mathrm{d}^{3}}{\mathrm{d\,}s^{3}}F_{\min }\left( \rho _{t},\rho
_{s}\right) $ in the sequel. Abbreviate $X_{t,s}:=\rho _{t}^{-1/2}\rho
_{t+s}\rho _{t}^{-1/2}$ ($>0$), and define $A_{t,s}$, $B_{t,s}$, and $C_{t,s}
$ by 
\begin{equation*}
\sqrt{X_{t,s+\varepsilon }}=\sqrt{X_{t,s}}+\varepsilon A_{t,s}+\varepsilon
^{2}B_{t,s}+\varepsilon ^{3}C_{t,s}+o\left( \varepsilon ^{3}\right) .
\end{equation*}%
They are determined by comparing the both ends of 
\begin{equation*}
X_{t,s}+\frac{\mathrm{d}X_{t,s}}{\mathrm{d\,}s}\varepsilon +\frac{1}{2}\frac{%
\mathrm{d}^{2}X_{t,s}}{\mathrm{d\,}s^{2}}\varepsilon ^{2}+\frac{1}{3!}\frac{%
\mathrm{d}^{3}X_{t,s}}{\mathrm{d\,}s^{3}}\varepsilon ^{3}+o\left(
\varepsilon ^{3}\right) =\left\{ \sqrt{X_{t,s}}+\varepsilon
A_{t,s}+\varepsilon ^{2}B_{t,s}+\varepsilon ^{3}C_{t,s}+o\left( \varepsilon
^{3}\right) \right\} ^{2},
\end{equation*}%
and thus 
\begin{eqnarray*}
A_{t,s}\sqrt{X_{t,s}}+\sqrt{X_{t,s}}A_{t,s} &=&\frac{\mathrm{d}X_{t,s}}{%
\mathrm{d\,}s}, \\
B_{t,s}\sqrt{X_{t,s}}+\sqrt{X_{t,s}}B_{t,s} &=&\frac{1}{2}\frac{\mathrm{d}%
^{2}X_{t,s}}{\mathrm{d\,}s^{2}}-\left( A_{t,s}\right) ^{2}, \\
C_{t,s}\sqrt{X_{t,s}}+\sqrt{X_{t,s}}C_{t,s} &=&\frac{1}{3!}\frac{\mathrm{d}%
^{3}X_{t,s}}{\mathrm{d\,}s^{3}}-\left( A_{t,s}B_{t,s}+B_{t,s}A_{t,s}\right) .
\end{eqnarray*}%
Since $\sqrt{X_{t,s}}>0$, the first equation determines $A_{t,s}$ uniquely,
and by Theorem\thinspace VII.2.12 of \cite{Bhatia}, 
\begin{equation*}
\sup_{t\leq s\leq t+\varepsilon }\left\Vert A_{t,s}\right\Vert \leq
\left\Vert X_{t,s}\right\Vert ^{-1/2}\left\Vert \frac{\mathrm{d}X_{t,s}}{%
\mathrm{d\,}s}\right\Vert .
\end{equation*}%
Then the second and third equality determines $B_{t,s}$ and $C_{t,s}$
uniquely, and 
\begin{eqnarray*}
\sup_{t\leq s\leq t+\varepsilon }\left\Vert B_{t,s}\right\Vert  &\leq
&f_{1}\left( \left\Vert \rho _{t}\right\Vert ,\left\Vert \frac{\mathrm{d}%
\rho _{s}}{\mathrm{d\,}s}\right\Vert ,\left\Vert \frac{\mathrm{d}^{2}\rho
_{s}}{\mathrm{d\,}s^{2}}\right\Vert \right) , \\
\sup_{t\leq s\leq t+\varepsilon }\left\Vert C_{t,s}\right\Vert  &\leq
&f_{2}\left( \left\Vert \rho _{t}\right\Vert ,\left\Vert \frac{\mathrm{d}%
\rho _{s}}{\mathrm{d\,}s}\right\Vert ,\left\Vert \frac{\mathrm{d}^{2}\rho
_{s}}{\mathrm{d\,}s^{2}}\right\Vert \right) ,
\end{eqnarray*}%
where $f_{1}$, $f_{2}$ is a continuous function. Therefore,

\begin{eqnarray*}
F_{\min }\left( \rho _{t},\rho _{t+\varepsilon }\right) &\leq &\frac{%
\varepsilon ^{2}}{2}\sup_{t\leq s\leq t+\varepsilon }\left\vert \frac{%
\mathrm{d}^{2}}{\mathrm{d\,}s^{2}}F_{\min }\left( \rho _{t},\rho _{s}\right)
\right\vert \\
&\leq &\frac{\varepsilon ^{2}}{2}\sup_{t\leq s\leq t+\varepsilon }\left\vert 
\mathrm{tr}\,B_{t,s}\rho _{t}\right\vert \\
&\leq &\frac{\varepsilon ^{2}}{2}\sup_{t\leq s\leq t+\varepsilon
}f_{1}\left( \left\Vert \rho _{t}\right\Vert ,\left\Vert \frac{\mathrm{d}%
\rho _{s}}{\mathrm{d\,}s}\right\Vert ,\left\Vert \frac{\mathrm{d}^{2}\rho
_{s}}{\mathrm{d\,}s^{2}}\right\Vert \right) ,
\end{eqnarray*}%
and 
\begin{eqnarray}
\left\vert F_{\min }\left( \rho _{t},\rho _{t+\varepsilon }\right) -\left( 1-%
\frac{\varepsilon ^{2}}{8}J_{t}^{R}\right) \right\vert &\leq &\frac{%
\varepsilon ^{3}}{6}\sup_{t\leq s\leq t+\varepsilon }\left\vert \frac{%
\mathrm{d}^{3}}{\mathrm{d\,}s^{3}}F_{\min }\left( \rho _{t},\rho _{s}\right)
\right\vert  \notag \\
&=&\frac{\varepsilon ^{3}}{6}\sup_{t\leq s\leq t+\varepsilon }\left\vert 
\mathrm{tr}\,C_{t,s}\rho _{t}\right\vert  \notag \\
&\leq &\frac{\varepsilon ^{3}}{6}\sup_{t\leq s\leq t+\varepsilon
}f_{2}\left( \left\Vert \rho _{t}\right\Vert ,\left\Vert \frac{\mathrm{d}%
\rho _{s}}{\mathrm{d\,}s}\right\Vert ,\left\Vert \frac{\mathrm{d}^{2}\rho
_{s}}{\mathrm{d\,}s^{2}}\right\Vert \right) .  \label{dF=JR}
\end{eqnarray}

Let us define 
\begin{equation}
F_{R}\left( \rho ,\sigma \right) :=\cos \left( \min_{C}\int_{C}\sqrt{%
J_{t}^{R}}\mathrm{d\,}t\right) ,  \label{FR-def}
\end{equation}%
\bigskip where minimization is taken over all the smooth paths connecting $%
\rho $ and $\sigma $. Obviously, $F_{R}\left( \rho ,\sigma \right) $
satisfies (N) and (M), and a triangle inequality (\ref{triangle-1}).
Therefore we have:

\begin{proposition}
\begin{equation*}
F_{\min }\left( \rho ,\sigma \right) \leq F_{R}\left( \rho ,\sigma \right)
\leq F\left( \rho ,\sigma \right) ,
\end{equation*}%
and $F_{\min }\left( \rho ,\sigma \right) $ does not coincide with $%
F_{R}\left( \rho ,\sigma \right) $.
\end{proposition}

\begin{proof}
Since $F_{R}\left( \rho ,\sigma \right) $ satisfies (N) and (M), the first
assertion is obtained by Theorem\thinspace \ref{th:fmin}. Since $F_{\min }$
does not satisfy (\ref{triangle-1}) while $F_{R}\left( \rho ,\sigma \right) $
does, we have the second assertion.
\end{proof}

\begin{theorem}
\begin{equation}
F_{\min }\left( \rho ,\sigma \right) =\cos \left( \min_{C}\frac{1}{2}\int_{C}%
\sqrt{J_{t}^{R}}\mathrm{d\,}t\right) ,  \label{Fmin=intJR}
\end{equation}%
\bigskip where minimization is taken over all the smooth paths $C=\left\{
\rho _{t}\right\} $ with $\rho _{0}=\rho $ and $\rho _{1}=\sigma $, with $%
\left[ L_{s}^{R},L_{t}^{R}\right] =0$, $0\leq \forall s,t\leq 1$.
\end{theorem}

\begin{proof}
$\left[ L_{s}^{R},L_{t}^{R}\right] =0$ and (\ref{rld-diagonal}) imply $%
L_{t}^{R}=NL_{t}N^{-1},\forall t$, where $L_{t}$ is a diagonal real matrix
and the column vectors $\left\vert \varphi _{x}\right\rangle $ ($x=1$,$%
\cdots $,$d^{\prime }$) of $N$ are normalized. Define  $P_{0}:=N^{-1}\rho
_{0}\left( N^{\dagger }\right) ^{-1}$, then  
\begin{equation*}
L_{0}^{R}\rho _{0}=NL_{0}P_{0}N^{\dagger }.
\end{equation*}%
Therefore,  $L_{0}P_{0}$ is Hermitian, and $P_{0}$ is also diagonal. 

Therefore, \  
\begin{equation*}
\rho _{t}=NP_{t}N^{\dagger },\,\,0\leq \forall t\leq 1,
\end{equation*}%
where $P_{t}$ is a real diagonal matrix. Writing the $x$th diagonal element
of $P_{t}$ as $p_{t}\left( x\right) $, $\sum_{x=1}^{d}p_{t}\left( x\right) =1
$, and thus $p_{t}\left( x\right) $ is a probability distribution over $%
\left\{ 1,\cdots ,d\right\} $. One can check that $L_{t}$ is the logarithmic
derivative of $p_{t}$, and that $J_{\rho _{t}}^{R}=J_{p_{t}}$ . Therefore, 
\begin{equation*}
\cos \left( \min_{C}\frac{1}{2}\int_{C}\sqrt{J_{\rho _{t}}^{R}}\mathrm{d\,}%
t\right) =\cos \left( \min_{C}\frac{1}{2}\int_{C}\sqrt{J_{p_{t}}}\mathrm{d\,}%
t\right) ,
\end{equation*}%
where the minimum is taken over all the smooth curves $C=\left\{
p_{t}\right\} $ in probability distributions connecting $p_{0}$ and $p_{1}$.
By (\ref{F=int-J}), the LHS equals $\cos ^{-1}F\left( p_{0},p_{1}\right) $. 

Define $\Phi \left( \delta _{x}\right) =\left\vert \varphi _{x}\right\rangle
\left\langle \varphi _{x}\right\vert $, then $\left( \Phi ,\left\{
p_{0},p_{1}\right\} \right) $ is a minimal reverse test of $\left\{ \rho
,\sigma \right\} $. Therefore, $\cos ^{-1}F\left( p_{0},p_{1}\right) $
equals $\cos ^{-1}F_{\min }\left( \rho ,\sigma \right) $, and we have the
assertion.\ 
\end{proof}

\begin{theorem}
Suppose that \ $F^{Q}\left( \rho ,\sigma \right) $ satisfies (M), (N), and 
\begin{equation*}
\cos ^{-1}F\left( \rho ,\sigma \right) \leq \cos ^{-1}F\left( \rho ,\tau
\right) +\cos ^{-1}F\left( \tau ,\sigma \right) .
\end{equation*}%
Then, 
\begin{equation*}
F_{R}\left( \rho ,\sigma \right) \leq F^{.Q}\left( \rho ,\sigma \right) \leq
F\,\left( \rho ,\sigma \right) .
\end{equation*}
\end{theorem}

\begin{proof}
We only have to show the lowerbound. Since any bounded and closed subset of
interior of $\mathcal{S}\left( \mathcal{H}\right) $ is compact, there is the
unique shortest Riemanian geodesic $C=\left\{ \rho _{t}\right\} $ with
respect to RLD Fisher information metric connecting any $\rho _{0}=\rho >0$
and $\rho _{1}=\sigma >0$ (see Theorem\thinspace 1.7.1 of \cite%
{riemannian-geometry}.) Then, by (\ref{triangle-1}), 
\begin{eqnarray*}
\cos ^{-1}F^{Q}\left( \rho ,\sigma \right)  &\leq &\sum_{k=0}^{1/\varepsilon
-1}\cos ^{-1}F^{Q}\left( \rho _{k\varepsilon },\rho _{\left( k+1\right)
\varepsilon }\right)  \\
&\leq &\sum_{k=0}^{1/\varepsilon -1}\cos ^{-1}F_{\min }\left( \rho
_{k\varepsilon },\rho _{\left( k+1\right) \varepsilon }\right) 
\end{eqnarray*}%
where the second line is due to Theorem\thinspace \ref{th:fmin}. By
elementary calculus, one can verify 
\begin{equation}
\cos ^{-1}F\leq \sqrt{2\left( 1-F\right) \left( 1+\frac{1}{6}\left(
1-F\right) \right) },\,\,(0\leq F\leq 1).  \label{cos-1}
\end{equation}%
Due to the smoothness of geodesic, (\ref{dF=JR}), and (\ref{cos-1}), letting 
\begin{equation*}
f_{i,t,s}:=f_{i}\left( \left\Vert \rho _{t}\right\Vert ,\left\Vert \frac{%
\mathrm{d}\rho _{s}}{\mathrm{d\,}s}\right\Vert ,\left\Vert \frac{\mathrm{d}%
^{2}\rho _{s}}{\mathrm{d\,}s^{2}}\right\Vert \right) ,\,\,\left(
i=1,2\right) ,
\end{equation*}%
we have 
\begin{eqnarray*}
&&\sum_{k=0}^{1/\varepsilon -1}\cos ^{-1}F_{\min }\left( \rho _{k\varepsilon
},\rho _{\left( k+1\right) \varepsilon }\right)  \\
&\leq &\sum_{k=0}^{1/\varepsilon -1}\sqrt{2\left( 1-F_{\min }\left( \rho
_{k\varepsilon },\rho _{\left( k+1\right) \varepsilon }\right) \right) }%
\sqrt{1+\frac{1}{6}\left( 1-F_{\min }\left( \rho _{k\varepsilon },\rho
_{\left( k+1\right) \varepsilon }\right) \right) } \\
&\leq &\sum_{k=0}^{1/\varepsilon -1}\sqrt{\frac{\varepsilon ^{2}}{4}%
J_{k\varepsilon }^{R}+\frac{1}{3}\sup_{k\varepsilon \leq s\leq \left(
k+1\right) \varepsilon }f_{1,t,s}\varepsilon ^{3}}\sqrt{1+\frac{\varepsilon
^{2}}{12}\sup_{k\varepsilon \leq s\leq \left( k+1\right) \varepsilon
}f_{2,t,s}} \\
&=&\sum_{k=0}^{1/\varepsilon -1}\frac{\varepsilon }{2}\left( J_{k\varepsilon
}^{R}\right) ^{\frac{1}{2}}\sqrt{\left( 1+\frac{4}{3}\frac{1}{\left(
J_{k\varepsilon }^{R}\right) ^{1/2}}\sup_{k\varepsilon \leq s\leq \left(
k+1\right) \varepsilon }f_{1,t,s}\varepsilon \right) \left( 1+\frac{%
\varepsilon ^{2}}{12}\sup_{k\varepsilon \leq s\leq \left( k+1\right)
\varepsilon }f_{2,t,s}\right) } \\
&\leq &\sum_{k=0}^{1/\varepsilon -1}\frac{\varepsilon }{2}\left(
J_{k\varepsilon }^{R}\right) ^{\frac{1}{2}}\sqrt{1+\varepsilon f_{3}},
\end{eqnarray*}%
where 
\begin{equation*}
f_{3}:=\sup_{t,s,t^{\prime },s^{\prime },u\in \left[ o,1\right] }\left( 1+%
\frac{4}{3}\frac{1}{\left( J_{u}^{R}\right) ^{1/2}}f_{1,t,s}\right) \left( 1+%
\frac{1}{12}f_{2,t^{\prime },s^{\prime }}\right) -1<\infty .
\end{equation*}%
Here, taking $\varepsilon \rightarrow 0$, we have the assertion.
\end{proof}

\section{Differential equation for shortest paths}

In this section, a differential equation satisfied by the geodesic, or the
path achieving the minimum in (\ref{FR-def}) is derived, supposing that $%
\rho _{t}$ is an invertible matrix. Let 
\begin{eqnarray*}
\mathfrak{L}\left( \rho _{t},\frac{\mathrm{d}\,\rho _{t}}{\mathrm{d}t}%
,\lambda _{t}\right)  &:&=\int_{C}\left\{ J_{t}^{R}-\lambda _{t}\left( 
\mathrm{tr}\,\rho _{t}-1\right) \right\} \mathrm{d\,}t \\
&=&\int_{C}\left\{ \mathrm{tr}\,\left( \frac{\mathrm{d}\,\rho _{t}}{\mathrm{d%
}t}\right) ^{2}\rho _{t}^{-1}-\lambda _{t}\left( \mathrm{tr}\,\rho
_{t}-1\right) \right\} \mathrm{d\,}t.
\end{eqnarray*}%
Taking $t$ proportional to arc length, this is equivalent to finding the
extremal  of $\int_{C}\left\{ \sqrt{J_{t}^{R}}-\lambda _{t}\left( \mathrm{tr}%
\,\rho _{t}-1\right) \right\} \mathrm{d\,}t$ \cite{riemannian-geometry}.
Then, letting $\left\{ X_{t}\right\} $ be an arbitrary smooth curve with $%
X_{0}=X_{1}=0$, 
\begin{eqnarray*}
&&\left. \frac{\mathrm{d}}{\mathrm{d}\varepsilon }\mathfrak{L}\left( \rho
_{t}+\varepsilon X_{t},\frac{\mathrm{d}\,\rho _{t}}{\mathrm{d}t}+\varepsilon 
\frac{\mathrm{d}X_{t}}{\mathrm{d}t},\lambda _{t}\right) \right\vert
_{\varepsilon =0} \\
&=&\int_{C}\left[ \mathrm{tr}\frac{\mathrm{d}X_{t}}{\mathrm{d}t}\left(
L_{t}^{R}+L_{t}^{R\dagger }\right) +\mathrm{tr}\,X_{t}\left\{
-L_{t}^{R\dagger }L_{t}^{R}+\lambda _{t}\right\} \right] \,\mathrm{d\,}t \\
&=&\left[ \mathrm{tr}X_{t}\left( L_{t}^{R}+L_{t}^{R\dagger }\right) \right]
_{t=0}^{1} \\
&&+\int_{C}\mathrm{tr}\,X_{t}\left\{ -\frac{\mathrm{d}}{\mathrm{d}t}\left(
L_{t}^{R}+L_{t}^{R\dagger }\right) -L_{t}^{R\dagger }L_{t}^{R}+\lambda
_{t}\right\} \,\mathrm{d\,}t
\end{eqnarray*}%
Hence, \ we have%
\begin{equation}
-\frac{\mathrm{d}}{\mathrm{d}t}\left( L_{t}^{R}+L_{t}^{R\dagger }\right)
-L_{t}^{R\dagger }L_{t}^{R}+\lambda _{t}=0,  \label{geodesic-1}
\end{equation}%
where the first identity is by $J_{t}^{R}=1$. In the sequel, the following
identity is used frequently. 
\begin{equation}
\rho _{t}L_{t}^{R\dagger }=\left( L_{t}^{R}\rho _{t}\right) ^{\dagger
}=\left( \frac{\mathrm{d}\,\rho _{t}}{\mathrm{d}t}\right) ^{\dagger
}=L_{t}^{R}\rho _{t}.  \label{rhoL=Lrho}
\end{equation}%
Multiplying $\rho $ and taking trace of both ends of (\ref{geodesic-1}), 
\begin{equation*}
-\mathrm{tr}\,\rho _{t}\frac{\mathrm{d}}{\mathrm{d}t}\left(
L_{t}^{R}+L_{t}^{R\dagger }\right) -1+\lambda _{t}=0,
\end{equation*}%
where we used $J_{t}^{R}=1$, (\ref{rhoL=Lrho}), and $\mathrm{tr}\,\rho
_{t}L_{t}^{R}=\mathrm{tr}\,\rho _{t}L_{t}^{R\dagger }=\mathrm{tr}\,\frac{%
\mathrm{d}\,\rho _{t}}{\mathrm{d}t}=0$. Observe 
\begin{eqnarray*}
\mathrm{tr}\,\rho _{t}\frac{\mathrm{d}}{\mathrm{d}t}\left(
L_{t}^{R}+L_{t}^{R\dagger }\right)  &=&\frac{\mathrm{d}}{\mathrm{d}t}\mathrm{%
tr}\,\rho _{t}\left( L_{t}^{R}+L_{t}^{R\dagger }\right) -\mathrm{tr}\,\frac{%
\mathrm{d}\,\rho _{t}}{\mathrm{d}t}\left( L_{t}^{R}+L_{t}^{R\dagger }\right) 
\\
&=&-\mathrm{tr}\,\left( \rho _{t}L_{t}^{R\dagger }L_{t}^{R}+L_{t}^{R}\rho
_{t}L_{t}^{R\dagger }\right) =-2.
\end{eqnarray*}%
Therefore, $\lambda _{t}=-1$, and we obtain 
\begin{equation}
\frac{\mathrm{d}}{\mathrm{d}t}\left( L_{t}^{R}+L_{t}^{R\dagger }\right)
+L_{t}^{R\dagger }L_{t}^{R}+1=0,  \label{geodesics-2}
\end{equation}%
which gives only determines time derivative of only Hermitian part of $%
L_{t}^{R}$. Obviously, we need another equation. By (\ref{rhoL=Lrho}), we
have 
\begin{equation*}
\frac{\mathrm{d}\rho _{t}}{\mathrm{d}t}L_{t}^{R\dagger }+\rho _{t}\frac{%
\mathrm{d}L_{t}^{R\dagger }}{\mathrm{d}t}=L_{t}^{R}\frac{\mathrm{d}\rho _{t}%
}{\mathrm{d}t}+\frac{\mathrm{d}L_{t}^{R}}{\mathrm{d}t}\rho _{t}.
\end{equation*}%
Observing 
\begin{equation*}
\frac{\mathrm{d}\rho _{t}}{\mathrm{d}t}L_{t}^{R\dagger }=L_{t}^{R}\rho
_{t}L_{t}^{R\dagger }=L_{t}^{R}\frac{\mathrm{d}\rho _{t}}{\mathrm{d}t},
\end{equation*}%
we obtain 
\begin{equation*}
\rho _{t}\frac{\mathrm{d}L_{t}^{R\dagger }}{\mathrm{d}t}=\frac{\mathrm{d}%
L_{t}^{R}}{\mathrm{d}t}\rho _{t},
\end{equation*}%
and 
\begin{equation}
\rho _{t}\frac{\mathrm{d}L_{t}^{R}}{\mathrm{d}t}+\frac{\mathrm{d}L_{t}^{R}}{%
\mathrm{d}t}\rho _{t}+\rho _{t}L_{t}^{R\dagger }L_{t}^{R}+\rho _{t}=0.
\label{geodesics-3}
\end{equation}%
(\ref{geodesics-3}) and 
\begin{equation*}
\frac{\mathrm{d}\rho _{t}}{\mathrm{d}t}=L_{t}^{R}\rho _{t}\,,
\end{equation*}%
determine time evolution of $\rho _{t}$.

The differential equation satisfied by the curve achieving minimum in (\ref%
{F=int-J}) is derived by applying these to commutative case, 
\begin{equation}
2\frac{\mathrm{d}l_{t}}{\mathrm{d}t}+\left( l_{t}\right) ^{2}+1=0,\,\frac{%
\mathrm{d}p_{t}}{\mathrm{d}t}=l_{t}\,p_{t}.  \label{classical-geodesic-2}
\end{equation}%
From these, the differential equation for the curve achieving minimum in (%
\ref{Fmin=intJR}) is derived as follows. Along the curve, we should have 
\begin{equation*}
L_{t}^{R}=NL_{t}N^{-1},\,\rho _{t}=N\left( D_{t}\right) ^{2}N^{\dagger },
\end{equation*}%
where $\left( D_{t}\right) ^{2}=\mathrm{diag}\,\left( p_{t}\left( 1\right)
,\cdots p_{t}\left( d\right) \right) $, $L_{t}=\mathrm{diag}\,\left(
l_{t}\left( 1\right) ,\cdots l_{t}\left( d\right) \right) $. \ Since $%
F_{\min }\left( \rho _{0},\rho _{1}\right) =F\left( p_{0},p_{1}\right) $, $%
\left( p_{t},l_{t}\right) $ should satisfy (\ref{classical-geodesic-2}).
Therefore, dynamics of $(\rho _{t},L_{t}^{R})$ is determined by 
\begin{equation}
2\frac{\mathrm{d}L_{t}^{R}}{\mathrm{d}t}+\left( L_{t}^{R}\right)
^{2}+1=0,\,\,\frac{\mathrm{d}\rho _{t}}{\mathrm{d}t}=L_{t}^{R}\rho _{t}\,.
\label{geodesic-Fmin}
\end{equation}

\section{Another quantum analogy of statistical distance}

Statistical distance $\Delta \left( p,q\right) $ is defined by 
\begin{equation*}
\Delta \left( p,q\right) :=\frac{1}{2}\left\Vert p-q\right\Vert _{1}=\frac{1%
}{2}\sum_{x\in \mathcal{X}}\left\vert p\left( x\right) -q\left( x\right)
\right\vert .
\end{equation*}%
Its frequently used quantum analogy is 
\begin{equation*}
\Delta \left( \rho ,\sigma \right) :=\frac{1}{2}\left\Vert \rho -\sigma
\right\Vert _{1},
\end{equation*}%
and also called statistical distance. Known facts about them are \cite{Fuchs}%
\cite{Matsumoto-0}\cite{Matsumoto-1}\cite{nielsen-chuang}\cite{Uhlmann3}:

\begin{itemize}
\item 
\begin{equation*}
\Delta \left( \rho ,\sigma \right) =\max_{M\text{:measurement}}\Delta \left(
M\left( \rho \right) ,M\left( \sigma \right) \right) .
\end{equation*}

\item (Monotonicity by CPTP maps) If $\Lambda $ is a CPTP map,%
\begin{equation*}
\Delta \left( \Lambda \left( \rho \right) ,\Lambda \left( \sigma \right)
\right) \leq \Delta \left( \rho ,\sigma \right)
\end{equation*}

\item (Joint convexity)%
\begin{equation*}
\Delta \left( \lambda \rho _{0}+\left( 1-\lambda \right) \rho _{1},\lambda
\sigma _{0}+\left( 1-\lambda \right) \sigma _{1}\right) \leq \lambda \Delta
\left( \rho _{0},\sigma _{0}\right) +\left( 1-\lambda \right) \Delta \left(
\rho _{1},\sigma _{1}\right)
\end{equation*}

\item (Triangle inequality)%
\begin{equation*}
\Delta \left( \rho ,\sigma \right) \leq \Delta \left( \rho ,\tau \right)
+\Delta \left( \tau ,\sigma \right)
\end{equation*}

\item With $\mathrm{D}\left( \rho ||\sigma \right) :=\mathrm{tr}\,\rho
\left( \ln \rho -\ln \sigma \right) $ being relative entropy, 
\begin{eqnarray}
1-F\left( \rho ,\sigma \right) &\leq &\Delta \left( \rho ,\sigma \right)
\leq \sqrt{1-F\left( \rho ,\sigma \right) ^{2}},  \label{1-F<D-1} \\
\mathrm{D}\left( \rho ||\sigma \right) &\geq &\frac{1}{2}\Delta \left( \rho
,\sigma \right) ^{2}.  \label{pinsker-1}
\end{eqnarray}
\end{itemize}

Here we introduce a new quantum analogue of statistical distance is:%
\begin{equation*}
\Delta _{\max }\left( \rho ,\sigma \right) :=\min_{\left( \Phi ,\left\{
p,q\right\} \right) \text{:reverse test of }\left\{ \rho ,\sigma \right\}
}\Delta \left( p,q\right) .
\end{equation*}

\begin{theorem}
Suppose $\Delta ^{Q}\left( \rho ,\sigma \right) $ satisfies monotonicity by
CPTP maps and $\Delta ^{Q}\left( p,q\right) =\Delta \left( p,q\right) $ for
any probability distributions $p$, $q$. Then%
\begin{equation*}
\Delta \left( \rho ,\sigma \right) \leq \Delta ^{Q}\left( \rho ,\sigma
\right) \leq \Delta _{\max }\left( \rho ,\sigma \right) .
\end{equation*}%
Also, \ $\Delta _{\max }$ is monotone by CPTP maps and $\Delta _{\min
}\left( p,q\right) =\Delta \left( p,q\right) $.
\end{theorem}

\begin{proof}
Almost parallel with the proof of Theorem\thinspace \ref{th:fmin}, thus
omitted.
\end{proof}

\begin{theorem}
Defining $\mathrm{D}^{R}\left( \rho ||\sigma \right) :=\mathrm{tr}\,\rho \ln 
\sqrt{\rho }\sigma ^{-1}\sqrt{\rho },$%
\begin{eqnarray}
\Delta _{\max }\left( \lambda \rho _{0}+\left( 1-\lambda \right) \rho
_{1},\lambda \sigma _{0}+\left( 1-\lambda \right) \sigma _{1}\right) &\leq
&\lambda \Delta _{\max }\left( \rho _{0},\sigma _{0}\right) +\left(
1-\lambda \right) \Delta _{\max }\left( \rho _{1},\sigma _{1}\right) ,
\label{joint-convex} \\
1-F_{\min }\left( \rho ,\sigma \right) &\leq &\Delta _{\max }\left( \rho
,\sigma \right) \leq \sqrt{1-F_{\min }\left( \rho ,\sigma \right) ^{2}},
\label{1-F<D} \\
\mathrm{D}^{R}\left( \rho ||\sigma \right) &\geq &\frac{1}{2}\Delta _{\max
}\left( \rho ,\sigma \right) ^{2}.  \label{pinsker}
\end{eqnarray}
\end{theorem}

\begin{proof}
The proof of (\ref{joint-convex}) is almost parallel with the one of
Theorem\thinspace \ref{th:strong-concave}, thus omitted. To prove the first
inequality of (\ref{1-F<D}), consider the optimal reverse test with $\Delta
\left( p,q\right) =\Delta _{\max }\left( \rho ,\sigma \right) $. Then, by (%
\ref{1-F<D-1}), we have 
\begin{equation*}
\Delta _{\max }\left( \rho ,\sigma \right) =\Delta \left( p,q\right) \geq
1-F\left( p,q\right) ,
\end{equation*}%
On the other hand, by definition of $F_{\min }$, $1-F\left( p,q\right) \geq
1-F_{\max }\left( \rho ,\sigma \right) $. After all, we have $\Delta _{\max
}\left( \rho ,\sigma \right) \geq 1-F_{\min }\left( \rho ,\sigma \right) $.
The second inequality of (\ref{1-F<D}) and (\ref{pinsker}) are proved almost
parallelly,recalling the following characterization of $\mathrm{D}^{R}\left(
\rho ||\sigma \right) $ (Theorem\thinspace 2.4 in \cite{Matsumoto-3}): 
\begin{equation*}
\mathrm{D}^{R}\left( \rho ||\sigma \right) =\min \mathrm{D}\left(
p||q\right) ,
\end{equation*}%
where minimum is taken over all the reverse test $\left( \Phi ,\left\{
p,q\right\} \right) $ of $\left\{ \rho ,\sigma \right\} $.
\end{proof}

It follows from (\ref{1-F<D}) and (\ref{F-pure}) that 
\begin{equation*}
\Delta _{\max }\left( \left\vert \varphi \right\rangle ,\left\vert \psi
\right\rangle \right) =1,
\end{equation*}%
for all $\left\vert \varphi \right\rangle $,$\left\vert \psi \right\rangle $%
. Another consequence of (\ref{1-F<D}) is:

\begin{proposition}
The triangle inequality for $\Delta _{\max }$ does not hold, i.e., there is $%
\rho $, $\sigma $, and $\tau $ with 
\begin{equation*}
\Delta _{\max }\left( \rho ,\sigma \right) >\Delta _{\max }\left( \rho ,\tau
\right) +\Delta _{\max }\left( \tau ,\sigma \right) .
\end{equation*}
\end{proposition}

\begin{proof}
Let $\rho $ and $\sigma $ be as of the proof of Proposition\thinspace \ref%
{prop:no-triangle}. Then, 
\begin{eqnarray*}
\Delta _{\max }\left( \rho ,\sigma \right) &=&1, \\
\Delta _{\max }\left( \rho ,\tau \right) +\Delta _{\max }\left( \tau ,\sigma
\right) &\leq &\sqrt{1-F_{\min }\left( \rho ,\tau \right) ^{2}}+\sqrt{%
1-F_{\min }\left( \tau ,\sigma \right) ^{2}} \\
&=&2\sqrt{1-\frac{1}{\left( \left\vert \cos \frac{\theta }{2}\right\vert
+\left\vert \sin \frac{\theta }{2}\right\vert \right) ^{2}}}.
\end{eqnarray*}%
Hence, making $\theta $ close enough to 0, we have the assertion.
\end{proof}

In general, $\Delta _{\max }\left( \rho ,\sigma \right) $ is hard to
compute. But when $\sigma =\left\vert \varphi \right\rangle \left\langle
\varphi \right\vert $, one can compute the number as follows. Any reverse
test $\left( \Phi ,\left\{ p,q\right\} \right) $ of $\left\{ \rho
,\left\vert \varphi \right\rangle \right\} $ is in the form of (\ref%
{reverse-pure}). Then, with $c:=\sum_{x\in \mathrm{supp}\,q\,}p\left(
x\right) $, 
\begin{equation*}
\left\Vert p-q\right\Vert _{1}\geq \frac{1}{2}\left\vert c-1\right\vert +%
\frac{1}{2}\left\vert \left( 1-c\right) -0\right\vert =1-c,
\end{equation*}%
where the inequality is due to monotonicity of $\left\Vert \cdot \right\Vert
_{1}$, and is achieved by $p\left( x\right) =cq\left( x\right) $( $x\in 
\mathrm{supp}\,q\,$). Therefore, to minimize $\left\Vert p-q\right\Vert _{1}$%
, one has to maximize $c$, which can be any positive number satisfying $\rho
-c\left\vert \varphi \right\rangle \left\langle \varphi \right\vert \geq 0$,
or equivalently, 
\begin{equation*}
1-c\rho ^{-1/2}\left\vert \varphi \right\rangle \left\langle \varphi
\right\vert \rho ^{-1/2}\geq 0.
\end{equation*}%
Therefore, 
\begin{equation*}
\Delta _{\max }\left( \rho ,\left\vert \varphi \right\rangle \right) =1-%
\frac{1}{\left\Vert \rho ^{-1/2}\left\vert \varphi \right\rangle \right\Vert
^{2}}=1-F\left( \rho ,\left\vert \varphi \right\rangle \right) ^{2}.
\end{equation*}%
Also, by the argument in Section\thinspace\ , we have the following
upperbound to $\Delta _{\max }\left( \rho ,\sigma \right) $:%
\begin{equation*}
\Delta _{\max }\left( \rho ,\sigma \right) \leq \Delta \left( M\left( \rho
\right) ,M\left( T\rho T\right) \right) \leq \Delta \left( \rho ,T\rho
T\right) ,
\end{equation*}%
where $T$ is as of (\ref{Tdef}), and $M$ is the projectors onto eigenspaces
of $T$. (Note the left most end is upperbounded by 
\begin{equation*}
\sqrt{1-F\left( \rho ,T\rho T\right) ^{2}}=\sqrt{1-F_{\min }\left( \rho
,\sigma \right) ^{2}},
\end{equation*}%
and gives a better bound than (\ref{1-F<D}).) 

\section{Discussions}

$F_{f}^{\min }$ introduced in this paper resembles $g$-quasi relative
entropy\thinspace \cite{Petz} 
\begin{equation*}
S_{g}\left( \rho ,\sigma \right) :=\mathrm{tr}\,\rho ^{1/2}g\left( L_{\sigma
}R_{\rho }^{-1}\right) \left( \rho ^{1/2}\right) ,
\end{equation*}%
where $g$ is operator convex function (thus, $-g$ is operator monotone ) on $%
[0,\infty )$, and 
\begin{equation*}
L_{\sigma }\left( A\right) =\sigma A,\,\,R_{\rho }\left( A\right) =A\rho .
\end{equation*}%
For example, with $g=1-x^{1/2}$, it gives rise to $S_{\alpha }\left( \rho
,\sigma \right) :=1-\mathrm{tr}\,\rho ^{1-\alpha }\sigma ^{\alpha }$. This
quantity satisfies 
\begin{equation*}
S_{\alpha }\left( \rho ,\sigma \right) \geq 1-F_{\alpha }^{\min }\left( \rho
,\sigma \right) ,
\end{equation*}%
where 
\begin{equation*}
F_{\alpha }^{\min }\left( \rho ,\sigma \right) :=\mathrm{tr}\,\rho
^{1/2}T^{2\left( 1-\alpha \right) }\rho ^{1/2}=\mathrm{tr}\,\rho
^{1/2}\left( \rho ^{-1/2}\sigma \rho ^{-1/2}\right) ^{1-\alpha }\rho ^{1/2},
\end{equation*}%
and the equality does not hold in general. 

It is interesting that both of the maximum of $F_{f}\left( p,q\right) $ are
achieved by minimal reverse tests. Some numerics suggests that this is not
the case for the minimum of $\Delta \left( p,q\right) $.

An open question is weather $F_{R}$ satisfies strong joint convexity or not.
\ Also, ,more explicit formula for $F_{R}$ and $\Delta _{\max }$ would be,
even for some special cases, of importance.

\end{document}